\documentclass[runningheads]{llncs}
\usepackage{graphicx}
\usepackage{hyperref, xcolor}
\usepackage{amsmath}
\usepackage{amssymb}
\newtheorem{assumption}{Assumption}

\usepackage{xcolor}

\makeatletter
\renewcommand{\@Opargbegintheorem}[4]{%
  #4\trivlist\item[\hskip\labelsep{#3#2\@thmcounterend}]}
\makeatother

\begin{document}
\title{Asymptotic utility of spectral anonymization \thanks{Submitted and accepted for publication in the proceedings of Privacy in Statistical Databases (PSD) 2024, held in Antibes, France, from September 25–27, 2024.}}
%
%
\author{Katariina Perkonoja\inst{1,2}\orcidID{0000-0002-3812-0871} \and
Joni Virta \inst{1}\orcidID{0000-0002-2150-2769}}
\authorrunning{K. Perkonoja and J. Virta}
%
\institute{Department of Mathematics and Statistics, University of Turku, Finland \and
Department of Computing, University of Turku, Finland \\ \email{\{kakype,jomivi\}@utu.fi}}
\maketitle              
\begin{abstract}

In the contemporary data landscape characterized by multi-source data collection and third-party sharing, ensuring individual privacy stands as a critical concern. While various anonymization methods exist, their utility preservation and privacy guarantees remain challenging to quantify. In this work, we address this gap by studying the utility and privacy of the spectral anonymization (SA) algorithm, particularly in an asymptotic framework. Unlike conventional anonymization methods that directly modify the original data, SA operates by perturbing the data in a spectral basis and subsequently reverting them to their original basis. Alongside the original version $\mathcal{P}$-SA, employing random permutation transformation, we introduce two novel SA variants: $\mathcal{J}$-spectral anonymization and $\mathcal{O}$-spectral anonymization, which employ sign-change and orthogonal matrix transformations, respectively. We show how well, under some practical assumptions, these SA algorithms preserve the first and second moments of the original data. Our results reveal, in particular, that the asymptotic efficiency of all three SA algorithms in covariance estimation is exactly 50\% when compared to the original data. To assess the applicability of these asymptotic results in practice, we conduct a simulation study with finite data and also evaluate the privacy protection offered by these algorithms using distance-based record linkage. Our research reveals that while no method exhibits clear superiority in finite-sample utility, $\mathcal{O}$-SA distinguishes itself for its exceptional privacy preservation, never producing identical records, albeit with increased computational complexity. Conversely, $\mathcal{P}$-SA emerges as a computationally efficient alternative, demonstrating unmatched efficiency in mean estimation.

\end{abstract}

\section{Introduction}\label{sec:intro}

Various anonymization techniques and privacy protocols, such as $k$-anonymity \cite{sweeney2002kano}, $l$-diversity \cite{machanavajjhala2007ldi}, $t$-closeness \cite{ninghui2007tcl}, and differential privacy \cite{dwork2006diffp}, have been proposed to protect data privacy. While these tools are conceptually simple, their deeper mathematical properties are often intractable, meaning that the related anonymization methods generally lack proven theoretical guarantees, particularly regarding utility preservation, although some attempts have been made, see, e.g., \cite{dunsche2022multivariate,xiao2019local,awan2019benefits,seeman2021exact}. As a result, decision-makers encounter uncertainty when determining acceptable privacy thresholds for data disclosure, compounded by the trade-off between privacy and utility. Therefore, there exists a pressing need for theoretical results addressing the utility preservation of anonymization techniques, alongside their privacy guarantees.

In this work we study the utility and privacy of the spectral anonymization (SA) algorithm of Lasko and Vinterbo \cite{lasko2009spectral} in the asymptotic framework where the sample size of $n$ of the observed data grows without bounds. Given an observed data matrix $X_n \in \mathbb{R}^{n \times p}$ with $n$ observations of $p$ variables, SA anonymizes it as follows: Letting $H_n := I_n - (1/n) 1_n 1_n'$ be the $n \times n$ centering matrix, we first find the singular value decompostion (SVD) of the centered data $H_n X_n = U_n D_n V_n'$, where the matrix of left singular vectors is $U_n = (u_{n1} \mid \cdots \mid u_{np}) \in \mathbb{R}^{n \times p}$. Note that the subscript $n$ in all quantities serves to remind that our viewpoint is asymptotic (i.e., $n \rightarrow \infty$) and that we are in fact dealing with sequences of data matrices ($X_1, X_2, \ldots $). Let then $P_{n1}, \ldots, P_{np} \in \mathbb{R}^{n \times n}$ be i.i.d. random permutation matrices sampled uniformly from the set $\mathcal{P}_n$ of all $n \times n$ permutation matrices. Denoting by $U_{n0} := (P_{n1} u_{n1} \mid \cdots \mid P_{np} u_{np})$ the matrix of permuted singular vectors, the spectral anonymization of $X_n$ is
\begin{align}\label{eq:p_sa}
    X_{n, \mathcal{P}} = U_{n0} D_n V_n' + 1_n \Bar{x}_n',
\end{align}
where $\Bar{x}_n := (1/n) X_n' 1_n$. Note that each column of $U_n$ is permuted independently, meaning that $U_{n0}$ is not, in general, a simple row permutation of $U_n$. Intuitively, SA uses a very simple form or perturbation (permuting observations), which by itself would not guarantee anonymity, but since it is applied in another basis (given by the right singular vectors $V_n$), anonymity in the \textit{original basis} is reached. Technically, any basis could be used, but the \textit{spectral basis} given by the SVD is particularly appealing as the matrix $U_{n}$ has uncorrelated columns, implying that, after returning to the original basis, the anonymized data has approximately the same first and second moments as the original data. As argued also by \cite{lasko2009spectral}, since many standard statistical methods, such as regression, discriminant analysis, support vector machines, $k$-means clustering, etc., are based on means and correlations, this preservation of resemblance (moments) can be seen, more or less, as equivalent to the preservation of general data utility.

While permuting is possibly the most direct form of perturbation one can apply in the spectral basis, it is not the only option, and as our first main contribution, we propose two novel forms of SA, called $\mathcal{J}$-spectral anonymization and $\mathcal{O}$-spectral anonymization. The original SA defined in \eqref{eq:p_sa} will in the sequel be denoted as $\mathcal{P}$-SA, to distinguish it from these variants. In $\mathcal{J}$-SA, the transformation matrices used to perturb $U_n$ are drawn uniformly from the set $\mathcal{J}_n$ of all $n \times n$ sign-change matrices (diagonal matrices with diagonal elements $\pm 1$) and in $\mathcal{O}$-SA these matrices are drawn from the Haar uniform distribution on the set $\mathcal{O}_n$ of all $n \times n$ orthogonal matrices \cite{stewart1980efficient}. The resulting spectral anonymizations of $X_n$ are denoted by $X_{n, \mathcal{J}}$ and $X_{n, \mathcal{O}}$, respectively. Figure \ref{fig:demo} illustrates the effect of the three perturbations in a simple bivariate scenario.

\begin{figure}[t]
\includegraphics[width=\textwidth]{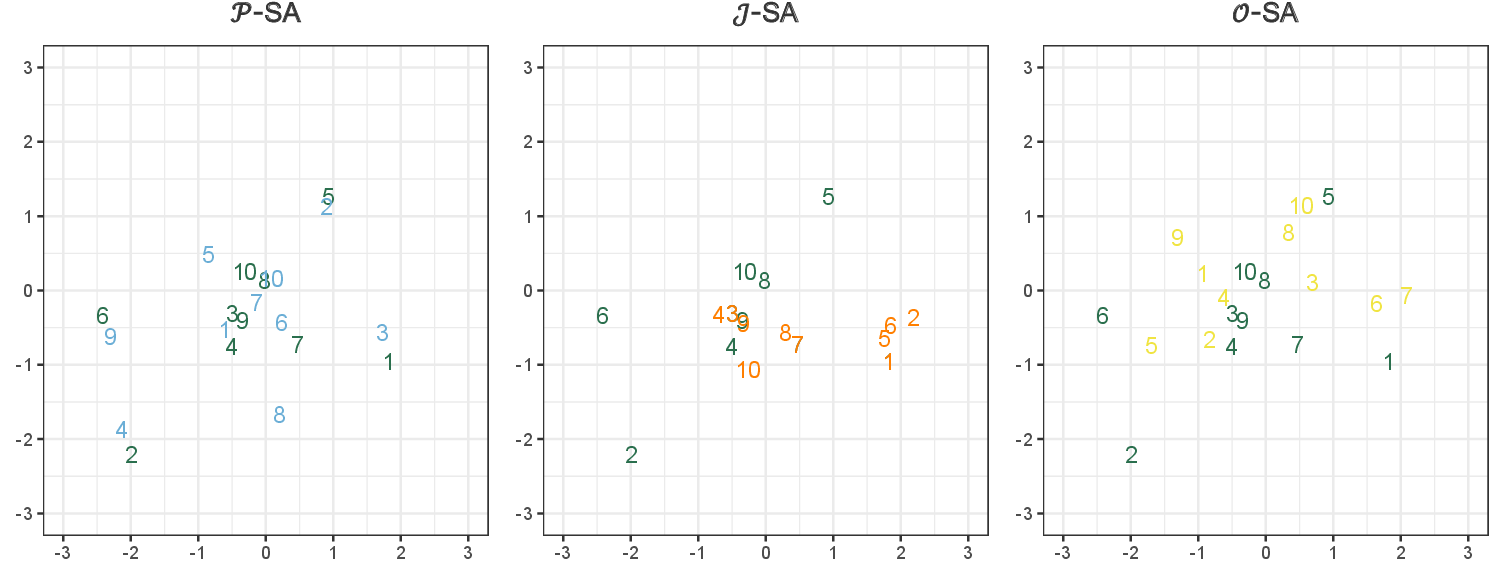}
\caption{Illustration of the effects of $\mathcal{P}$-SA (blue), $\mathcal{J}$-SA (orange) and $\mathcal{O}$-SA (yellow) to the original data (green/darkest), with numerical labels corresponding to the row indices of original data, in a scenario with $n = 10$ and $p = 2$. $\mathcal{P}$-SA and $\mathcal{J}$-SA occasionally result in unwanted overlap, wherein a row in the anonymized data matches another row in the original dataset. In the case of $\mathcal{J}$-SA, this occurrence is due to coincidental alignment of signs with those in the original singular vectors, causing the observation to be duplicated (overlap of same indices), whereas for $\mathcal{P}$-SA, random permutation of values may cause another row to align with one in the original dataset (overlap of different indices). Conversely, $\mathcal{O}$-SA induces arbitrary rotations in the spectral space that prevent exact matches altogether with probability 1.}
\label{fig:demo}
\end{figure}

As our second main contribution, we investigate the earlier observation, stated also by Lasko and Vinterbo \cite{lasko2009spectral}, that $\mathcal{P}$-SA approximately preserves the means, variances, covariances, and linear correlations of the original data $X_n$. In this work, we quantify this statement in an asymptotic framework. {While SA itself makes no distributional assumptions and is applicable to any continuous data, we derive our results, for simplicity, under the assumption that the rows of $X_n$ are sampled i.i.d. from the normal distribution $\mathcal{N}_p(\mu, \Sigma)$ for some fixed $\mu \in \mathbb{R}^p$ and positive definite covariance matrix $\Sigma$. Extensions to other distributional assumptions are investigated through simulations in Section \ref{sec:simu} and discussed further in Section \ref{sec:conclusion}. As the multivariate normal distribution is entirely determined by its means and covariances, evaluating the resemblance of anonymized data with respect to these moments is equivalent to assessing general data utility in the present scenario.

As our third main contribution, we conduct a simulation study that investigates how well the three forms of SA preserve utility and privacy under a wide range of data scenarios and sample sizes. We are not aware of an equally extensive benchmarking study having been carried out for the original SA of \cite{lasko2009spectral}, although \cite{kundu2019privacy} compared the prediction utility of SA to that of a private Bayesian factor model. Also, closely related to this, \cite{calvino2017fa} conducted a simulation study assessing utility loss when anonymizing data using another form of spectral anonymization based on factor analysis and principal component analysis.

In Section \ref{sec:asymp} we present our results on the asymptotic utility of the three forms of SA. The corresponding finite-sample behavior is studied using simulations in Section \ref{sec:finite}, whereas finite-sample privacy preservation is investigated in Section \ref{sec:priv}. In Section \ref{sec:conclusion} we conclude with discussion. The proofs are presented in Appendix \ref{sec:appen_proofs} and Appendix \ref{sec:appen_fig} contains additional simulation figures.

\section{Asymptotic utility} \label{sec:asymp}

We divide our theoretical results in two parts, investigating first how well the SA-methods preserve the first moment (mean vector) and then doing the same for the second moment (covariance matrix). Recall from Section \ref{sec:intro} that we assume the rows of $X_n$ to be i.i.d. from $\mathcal{N}_p(\mu, \Sigma)$. Throughout this section we impose the following technical condition on the covariance matrix $\Sigma$, which guarantees that the singular spaces of $H_n X_n$ are asymptotically identifiable up to sign. This assumption is mild and practically always satisfied for continuous real data.

\begin{assumption}\label{assu:eigenvalues}
    The eigenvalues of $\Sigma$ are distinct.
\end{assumption}

The classical central limit theorem gives for $\Bar{x}_n$ (the mean of the original data) the following limiting distribution as $n \rightarrow \infty$:
\begin{align*}
    \sqrt{n} (\Bar{x}_n - \mu) \rightsquigarrow \mathcal{N}_p(0, \Sigma),
\end{align*}
where $\rightsquigarrow$ denotes convergence in distribution. As our first result in this section, we derive the limiting distributions of the sample means $\Bar{x}_{n, \mathcal{P}}, \Bar{x}_{n, \mathcal{J}}, \Bar{x}_{n, \mathcal{O}}$ of the different spectral anonymzations of $X_n$.

\begin{theorem}\label{theo:means}
    Under Assumption \ref{assu:eigenvalues}, we have the following, as $n \rightarrow \infty$.
    \begin{align*}
        \sqrt{n} ( \Bar{x}_{n, \mathcal{P}} - \mu ) = \sqrt{n} (\Bar{x}_n - \mu) &\rightsquigarrow \mathcal{N}_p(0, \Sigma) \\
        \sqrt{n} ( \Bar{x}_{n, \mathcal{J}} - \mu ) &\rightsquigarrow \mathcal{N}_p(0, 2 \Sigma) \\
        \sqrt{n} ( \Bar{x}_{n, \mathcal{O}} - \mu ) &\rightsquigarrow \mathcal{N}_p(0, 2 \Sigma)
    \end{align*}
\end{theorem}

Two estimators of the same parameter can be compared based on the ``magnitudes'' of their limiting covariance matrices (``smaller'' covariance matrix implies more accurate estimator). Thus, Theorem \ref{theo:means} shows that $\mathcal{P}$-SA yields mean estimation efficiency equivalent to the original data (in fact, the two estimators are always exactly the same). Whereas $\mathcal{J}$-SA and $\mathcal{O}$-SA incur a cost in efficiency. Interestingly, this extra cost has a very simple form and the asymptotic efficiency of these two methods compared to mean estimation from the original data turn out to be exactly one half (due to the factor 2 in front of $\Sigma$).

We next conduct an analogous study of covariance matrix estimation. We denote the sample covariance matrices produced by the three forms of SA as $S_{n, \mathcal{P}}, S_{n, \mathcal{J}}, S_{n, \mathcal{O}}$. As a baseline, Theorem 1 in \cite{tyler1982radial}, shows that the limiting distribution of $S_n$, the sample covariance matrix of the original data, is 
\begin{align*}
        \sqrt{n} \mathrm{vec}( S_{n} - \Sigma ) \rightsquigarrow \mathcal{N}_{p^2}(0, (\Sigma^{1/2} \otimes \Sigma^{1/2}) (I_{p^2} + K_{p, p}) (\Sigma^{1/2} \otimes \Sigma^{1/2})' ),
\end{align*}
where $\mathrm{vec}(\cdot)$ denotes column-wise vectorization, $\otimes$ is the Kronecker product, $K_{p, p}$ is the $(p, p)$-commutation matrix \cite{kollo2005advanced} and the square root matrix $\Sigma^{1/2} := O \Lambda^{1/2}$ is computed using the eigendecomposition $\Sigma = O \Lambda O'$.


\begin{theorem}\label{theo:covariances}
    Under Assumption \ref{assu:eigenvalues}, we have for $\mathcal{A} \in \{ \mathcal{P}, \mathcal{J}, \mathcal{O} \}$ the following, as $n \rightarrow \infty$.
    \begin{align*}
        \sqrt{n} \mathrm{vec}( S_{n, \mathcal{A}} - \Sigma ) \rightsquigarrow \mathcal{N}_{p^2}(0, (\Sigma^{1/2} \otimes \Sigma^{1/2}) (2 I_{p^2} + 2 K_{p,p} - 2 V_p ) (\Sigma^{1/2} \otimes \Sigma^{1/2})'),
    \end{align*}
    where $V_p := \mathrm{diag} \{ \mathrm{vec}(I_p) \}$.
\end{theorem}

Theorem \ref{theo:covariances} offers us two insights: (a) All three forms of SA are equally efficient in covariance estimation. (b) Assume next that $O = I_p$ in the eigendecomposition $\Sigma = O \Lambda O'$ (since the rotation $O$ only fixes the coordinate system in which the data is observed, this choice does not limit our interpretations). Then, for instance when $p = 2$, the limiting covariance matrices of $S_n$ and $S_{n, \mathcal{A}}$ take the forms:
\begin{align}\label{eq:two_cov_matrices}
    \begin{pmatrix}
        2 \lambda_1^2 & 0 & 0 & 0 \\
        0 & \lambda_1 \lambda_2 & \lambda_1 \lambda_2 & 0 \\
        0 & \lambda_1 \lambda_2 & \lambda_1 \lambda_2 & 0 \\
        0 & 0 & 0 & 2 \lambda_2^2
    \end{pmatrix} \quad \mbox{and} \quad 
    \begin{pmatrix}
        2 \lambda_1^2 & 0 & 0 & 0 \\
        0 & 2 \lambda_1 \lambda_2 & 2 \lambda_1 \lambda_2 & 0 \\
        0 & 2 \lambda_1 \lambda_2 & 2 \lambda_1 \lambda_2 & 0 \\
        0 & 0 & 0 & 2 \lambda_2^2
    \end{pmatrix},
\end{align}
respectively, where $\lambda_1, \lambda_2$ are the marginal variances of the data. We observe that the elements $(1, 1), (4, 4)$, giving the asymptotic variances of the sample variances, are equal in the two matrices, meaning that the estimation of variances is equally efficient in original data and in SA-generated data. Whereas, the elements  $(2, 2), (2, 3), (3, 2), (3, 3)$ in the matrices \eqref{eq:two_cov_matrices}, corresponding to the asymptotic variances and covariances of the sample covariances are two times larger for SA than for the original data. Thus, we conclude that SA makes the estimation of the cross-terms in the covariance matrix $\Sigma$ more difficult. This is entirely natural considering that each column of $U_n$ is in SA mixed independently of each other, compromising our ability to detect interactions (sample covariances) between the columns, but leaving the marginal distribution (sample variances) of each column intact. For arbitrary rotation $O$, the equivalent interpretation holds, only this time in the coordinate system specified by $O$.


\section{Simulations} \label{sec:simu}

To complement the earlier asymptotic results, in this section we study the finite-sample utility of $\mathcal{P}$-SA, $\mathcal{J}$-SA and $\mathcal{O}$-SA through a simulation study. As discussed in Section \ref{sec:asymp}, achieving the exact limiting distribution necessitates an infinite sample size. Therefore, our focus in Section \ref{sec:finite} is to evaluate how good of an approximation the limiting distribution is for finite sample sizes. Rather than directly comparing empirical distributions to the limiting distribution, we examine the similarity between the empirical covariance matrices and their limiting counterparts. Moreover, since $\mathcal{P}$-SA, $\mathcal{J}$-SA and $\mathcal{O}$-SA are designed for data anonymization, we evaluate their capability in protecting privacy in Section \ref{sec:priv} by employing distance-based record linkage. The source code used for conducting the simulation study presented in this work is available on GitLab \url{https://gitlab.utu.fi/kakype/asymptotic-utility-of-spectral-anonymization.git}.

\begin{figure}[t]
\includegraphics[width=\textwidth]{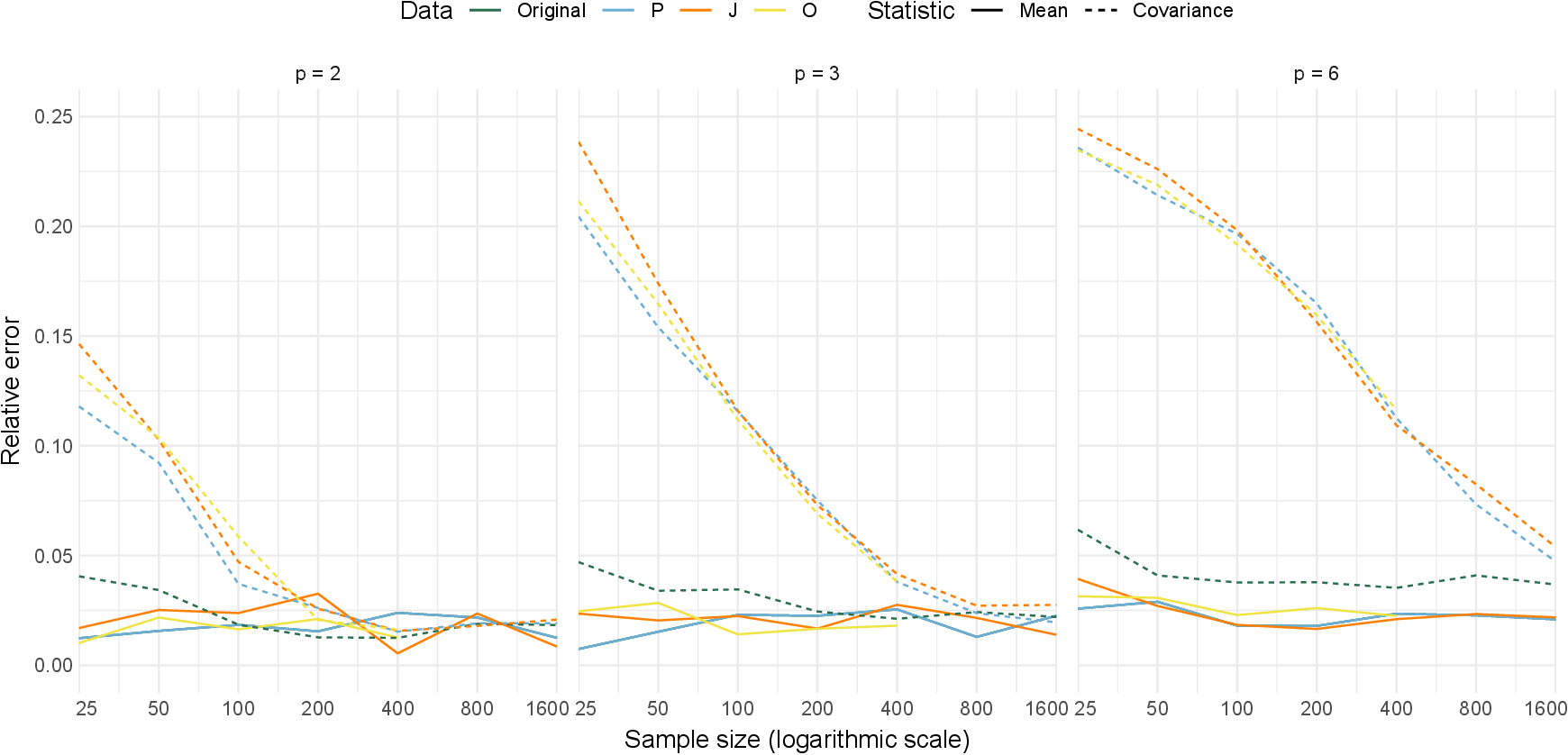}
\caption{The relative error of the empirical covariance matrices of sample mean (solid line) and sample covariance (dotted line) compared to their asymptotic covariance matrices for $\mathcal{P}$-SA (blue), $\mathcal{J}$-SA (orange), and $\mathcal{O}$-SA (yellow), when original data (green/darkest) is sampled from a normal distribution meeting the Assumption \ref{assu:eigenvalues}. Sample size is presented on a logarithmic scale.}
\label{fig:norm_met}
\end{figure}

\subsection{Finite-sample utility} \label{sec:finite}

We explore the impact of finite data on convergence, considering sample sizes $n = \{25, 50, 100, 200, 400, 800, 1600\}$ and number of variables $p = \{2, 3, 6\}$, reflecting plausible real-world scenarios. For each combination of $n$ and $p$, we sample original data $X_{n1}$ from $\mathcal{N}_p(\mu, \Sigma)$, where $\mu = (3, 3, \dots, 3) \in \mathbb{R}^p$ and $\Sigma$ represents a diagonal matrix with elements $(p, p-1, \dots, 1)$ to fulfill Assumption~\ref{assu:eigenvalues}. Additionally, we explore the consequences of deviating from normality and Assumption \ref{assu:eigenvalues} by generating three alternative datasets: $X_{n2}$, $X_{n3}$ and $X_{n4}$, whose rows are sampled i.i.d. from  $\mathcal{N}_p(\mu, I_p)$, $Poisson(\lambda_1)$ and $Poisson(\lambda_2)$, respectively, where $\lambda_1 = (p, p-1, \dots, 1)$, $\lambda_2 = (1, 1, \dots, 1) \in \mathbb{R}^p$ and $Poisson(c)$ denotes a distribution with indepedent $Poisson(c_k)$-distributed marginals. Note that, by the equivariance properties of the mean and covariance matrix, restricting our attention to data with uncorrelated variables is without loss of generality.

Each dataset was sampled $M = 10 000$ times and anonymized using $\mathcal{P}$-SA, $\mathcal{J}$-SA and $\mathcal{O}$-SA, respectively. However, due to the $O(n^3)$ complexity of $\mathcal{O}$-SA, we opted to perform this form of anonymization only for sample sizes $n \leq 400$. Subsequently, empirical values were computed according to the left-hand sides of Theorems \ref{theo:means} and \ref{theo:covariances} and the relative errors (RE) of their covariance matrices across the $M$ datasets were then calculated. By our theoretical results, these values approach zero as $n \rightarrow \infty$. For instance, in the case of sample mean (Theorem \ref{theo:means}) and $\mathcal{P}$-SA, the calculated value is 
$$
\text{RE}_{\Bar{x}_{n, \mathcal{P}}} = \frac{\| \widehat{Cov}_M(\sqrt{n} (\Bar{x}_{n, \mathcal{P}} - \mu))- \Sigma \|_{F}}{\|\Sigma\|_{F}},
$$ where $\mu$ and $\Sigma$ are the mean and covariance parameters derived from the aforementioned data distributions, and $\|\cdot\|_{F}$ represents the Frobenius norm. Figures~\ref{fig:norm_met} and \ref{fig:norm_viol} illustrate the results for normally distributed data, with corresponding representations for Poisson distribution found in Appendix \ref{sec:appen_fig}.

As predicted by Theorems \ref{theo:means} and \ref{theo:covariances}, all curves in Figure \ref{fig:norm_met} approach zero when $n$ grows (notably, even with small sample sizes, the convergence for the sample mean is eminent). There are no discernible differences in the convergence among $\mathcal{P}$-SA, $\mathcal{J}$-SA, and $\mathcal{O}$-SA towards their asymptotic covariance matrices. The relative error for the sample covariance depends on $p$, higher dimensions requiring a larger sample size for the RE to converge. When Assumption \ref{assu:eigenvalues} is violated, Figure~\ref{fig:norm_viol} indicates that the empirical covariance matrix of the sample covariance matrix converges to a different constant than outlined in Theorem~\ref{theo:covariances}. This implies that Assumption \ref{assu:eigenvalues} is necessary for our results to hold. Similar results were obtained regarding deviation from the normality assumption (Appendix \ref{sec:appen_fig}), where the sample covariance curves, even for the original data, failed to converge to zero. Hence, the results in Section \ref{sec:asymp} are not guaranteed to hold if the normality assumption is discarded, see Section \ref{sec:conclusion} for further discussion of this point.

\begin{figure}[t]
\includegraphics[width=\textwidth]{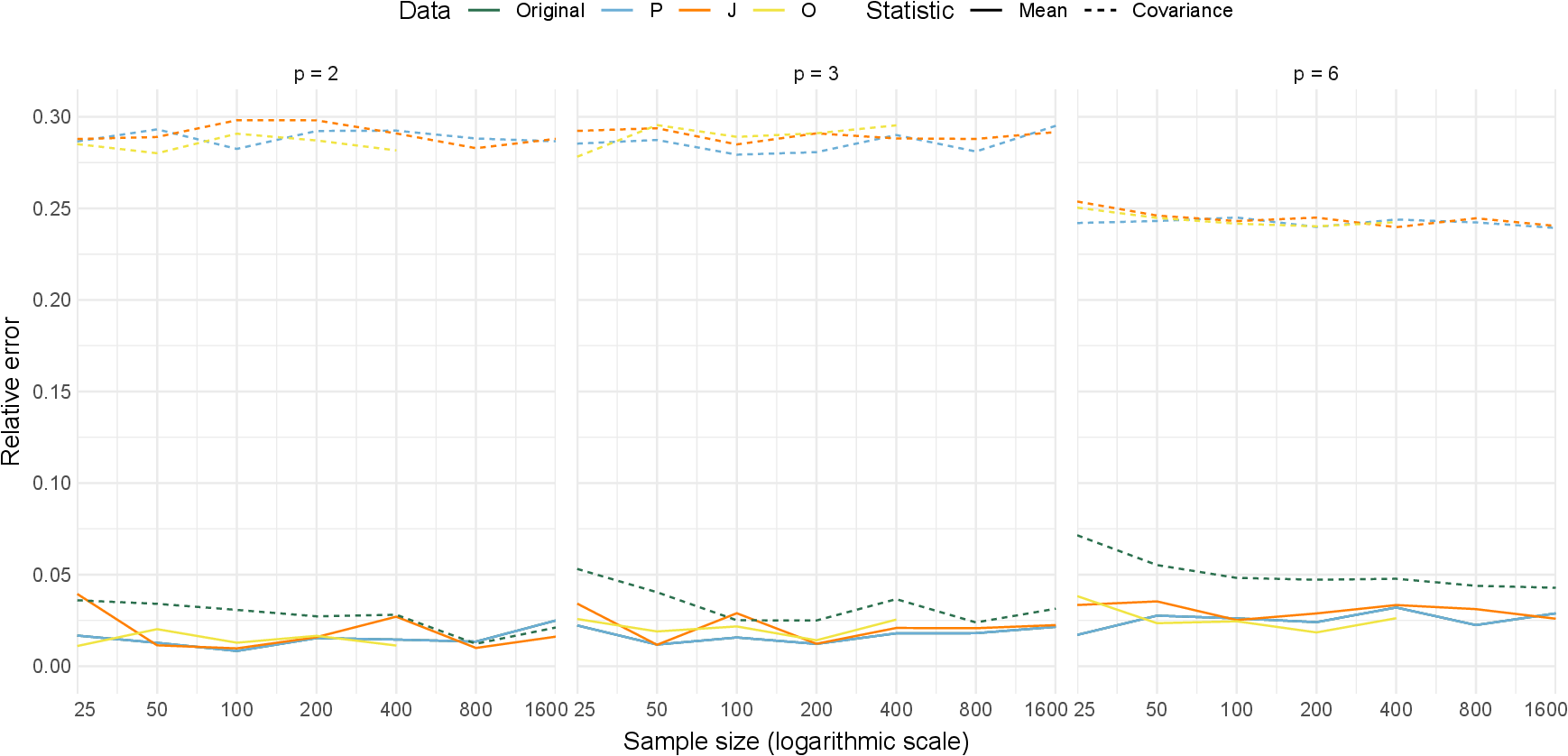}
\caption{The relative error of the empirical covariance matrices of sample mean (solid line) and sample covariance (dotted line) compared to their asymptotic covariance matrices for $\mathcal{P}$-SA (blue), $\mathcal{J}$-SA (orange), and $\mathcal{O}$-SA (yellow), when original data (green/darkest) is sampled from a normal distribution violating the Assumption \ref{assu:eigenvalues}. Sample size is presented on a logarithmic scale.}
\label{fig:norm_viol}
\end{figure}

\subsection{Privacy} \label{sec:priv}

In addition to assessing empirical convergence with finite data, we wanted to evaluate the privacy provided by $\mathcal{P}$-SA, $\mathcal{J}$-SA and $\mathcal{O}$-SA. We used distance-based record linkage, akin to the approach outlined in \cite{domingo2004disc},  with the distinction that all variables were utilized to compute the shortest Euclidean distance (EUC) between an anonymized record and any original record. The mean distance across all $m = 1, \dots, M$ simulated datasets was then calculated as follows: 
$$
\text{EUC}_{\mathcal{A}} = \frac{1}{M}\sum_{m=1}^{M}\frac{1}{n}\sum_{i=1}^{n} \min_{j} \|x_{(m),i,\mathcal{A}} - x_{(m), j}\|_2,
$$ 
where $x_{(m),i,\mathcal{A}}$ denotes the $i$th observation vector in the $m$th anonymized data and $x_{(m),i}$ denotes the equivalent for the $m$th original data. The distances were computed to the unstandardized data, aligning with the common practice of publishing datasets in their original scale. This anticipates potential linkage attempts by adversaries with external data sources. $\text{EUC}_{\mathcal{A}}$ embodies the privacy-utility trade-off, where values closer to zero indicate better utility to the original data but also imply higher privacy risk, with a distance of zero denoting essentially identical datasets. Consequently, higher values signify stronger privacy protection. Results of this comparison for normally distributed data meeting Assumption \ref{assu:eigenvalues} are depicted in Figure \ref{fig:euc}.

\begin{figure}[t]
\includegraphics[width=\textwidth]{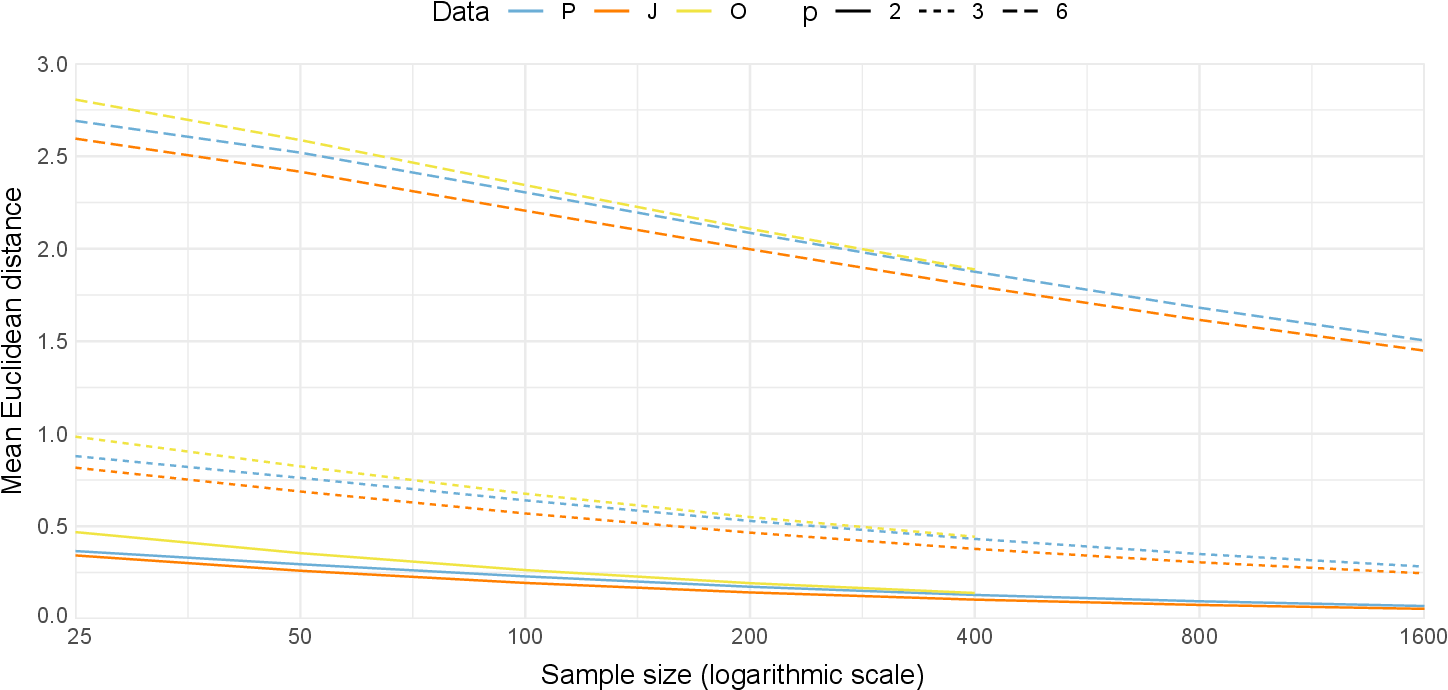}
\caption{$\text{EUC}_{\mathcal{A}}$ when the original data is sampled from normal distribution meeting Assumption \ref{assu:eigenvalues}. The y-axis illustrates the mean distance $\text{EUC}_{\mathcal{A}}$ between records in anonymized data and any record in the original data across all datasets, while the x-axis denotes the sample size on a logarithmic scale. The number of variables $p$ is distinguished by different linetypes and each anonymization approach is represented by a unique color.}
\label{fig:euc}
\end{figure}

While all SAs demonstrated comparable performance, $\mathcal{O}$-SA exhibited the best privacy protection, consistently yielding higher distances compared to $\mathcal{P}$-SA and $\mathcal{J}$-SA, as illustrated in Figure \ref{fig:euc}. These findings remained consistent across all simulation settings (results not shown). However, $\text{EUC}_{\mathcal{A}}$ is never precisely zero, as this would imply essentially complete equivalence between the anonymized and original data, which is highly unlikely for the given SAs. Yet, these algorithms can produce some matches by chance, as demonstrated in Figure \ref{fig:demo}, and for any individual, a direct match would constitute a privacy violation. Therefore, we supplemented our analysis by evaluating the proportion of matches, i.e., correctly linked records, relative to the sample size across all simulated datasets $m = 1, \dots, M$:
$$
\text{Matches}_{m,\mathcal{A}} = \frac{1}{n}\sum_{i=1}^{n} \begin{cases} 1 & \text{if } \min_{j} \|x_{(m), i,\mathcal{A}} - x_{(m),j}\|_2 < \delta, \\ 0 & \text{otherwise}, \end{cases}
$$
where the tolerance $\delta$ was set to $10^{-6}$. Results for normally distributed data meeting Assumption \ref{assu:eigenvalues} are presented in Figure \ref{fig:match}, with a lower proportion indicating better privacy protection (fewer matches).

\begin{figure}[h]
\includegraphics[width=\textwidth]{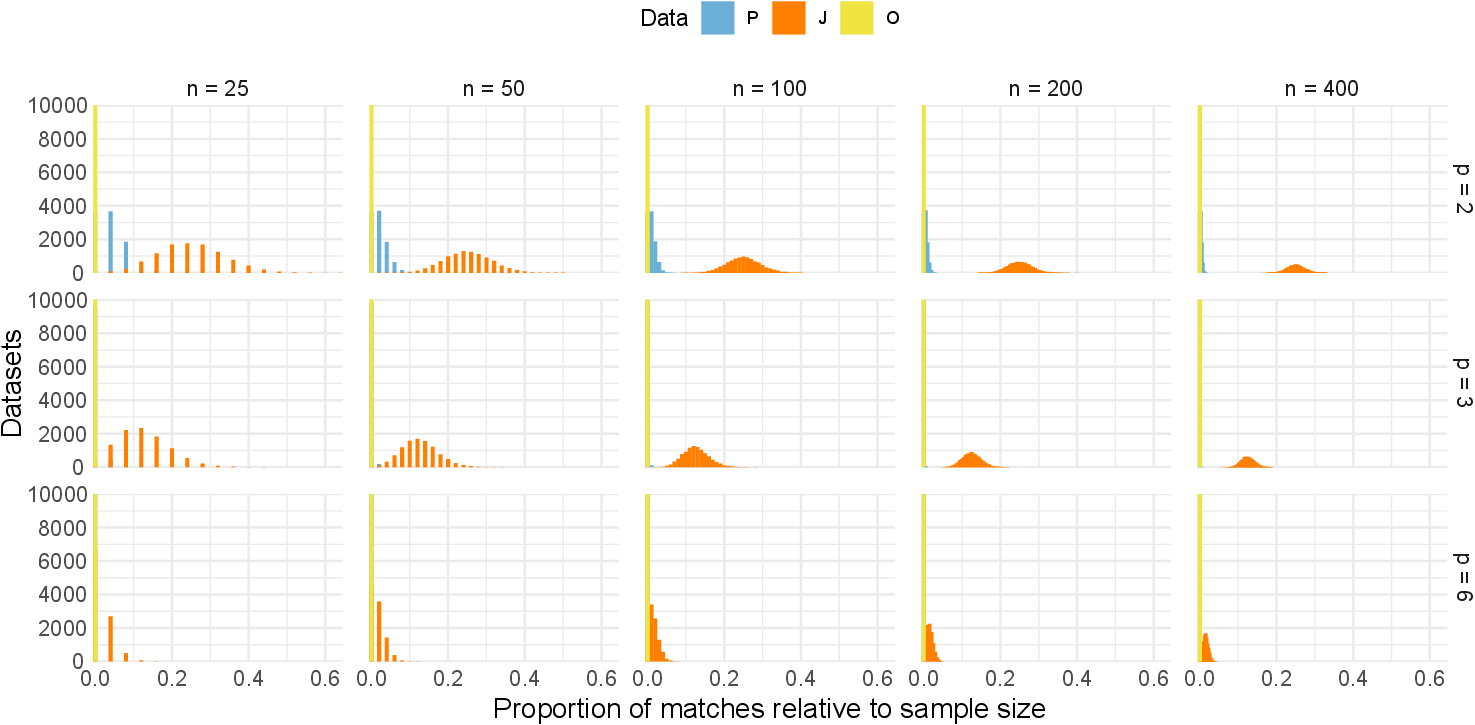}
\caption{Histograms of match proportions when the original data is sampled from normal distribution meeting Assumption \ref{assu:eigenvalues}. The y-axis represents the number of simulated datasets while the x-axis indicates the proportion of matches relative to the sample size, with different SAs represented by distinct colors. Sample sizes $n > 400$ have been omitted here, as they introduce no new information.}
\label{fig:match}
\end{figure}

$\mathcal{O}$-SA did not produce any matches, as illustrated in Figure \ref{fig:match}, with a proportion of zero observed across all datasets. Conversely, both $\mathcal{P}$-SA and $\mathcal{J}$-SA approaches resulted in matches, with the former indicating a lower frequency compared to the latter. However, as $p$ increased, the number of matches decreased, ultimately leading to no matches for $\mathcal{P}$-SA and approximately rate of 0.03 for $\mathcal{J}$-SA. This trend remained consistent across all simulation scenarios, with the Poisson distribution exhibiting slightly higher match proportion than the normal distribution, and deviations from Assumption 1 having only a marginal effect (results not shown).

\section{Conclusion}\label{sec:conclusion}

Based on our findings with finite data, none of the three methods demonstrated clear superiority in utility convergence. However, concerning privacy, $\mathcal{O}$-SA outperformed $\mathcal{P}$-SA and $\mathcal{J}$-SA by never generating identical records. Hence, we advocate for employing $\mathcal{O}$-SA for data anonymization, notwithstanding its computational intensity compared to $\mathcal{P}$-SA and $\mathcal{J}$-SA. For a dataset size of $n = 1000$ and $p=6$, the runtime for $\mathcal{O}$-SA is on average 2.44 seconds, while for $\mathcal{P}$-SA and $\mathcal{J}$-SA, it is 0.002 and 0.001 seconds, respectively (using AMD Ryzen 5 5600X). The complexity of SVD is the same as for principal component analysis, meaning that the method scales well. In practical scenarios, data anonymization is typically a one-time operation, mitigating concerns over computational overhead, and all SA algorithms can be applied to larger datasets (both in terms of $n$ and $p$) provided the data is continuous. However, ensuring asymptotic utility (Section \ref{sec:asymp}) necessitates additional assumptions. Should computational efficiency become paramount, $\mathcal{P}$-SA emerges as the secondary choice 
(no matches in our simulations with $p=6$). In case privacy is given secondary priority, then $\mathcal{P}$-SA offers superior mean estimation efficiency to its two competitors (Theorem \ref{theo:means}) and should be used. Nonetheless, further research is warranted to evaluate these SA algorithms using empirical real-world data and utility scenarios.



%

Our simulations showed that Assumption \ref{assu:eigenvalues} cannot be dispensed with. SA is based on perturbing the singular vectors of $H_n X_n$ and, intuitively, the role of Assumption~\ref{assu:eigenvalues} is to ensure that these vectors are asymptotically identifiable. Similarly, application to Poisson-distributed data revealed that our conclusions do not directly extend outside of normality. However, we expect that equivalent results could be obtained also under other distributional assumptions, with suitable techniques of proof. And while we chose to derive our theoretical results under the normality assumption due to the ubiquitousness of the Gaussian distribution, the SA-method itself remains applicable also outside of normal data.


A prospective research direction would be to conduct analogous studies for competing methods, enabling utility and privacy comparisons between the methods. E.g., data shuffling \cite{muralidhar2006data} and correlated noise addition \cite{shlomo2008protection} have been shown to preserve second-order statistics. Unlike SA, data shuffling and noise addition use distributional assumptions as part of their data generation and their impact on the methods' utility could be quantified with results such as Theorem~\ref{theo:means} and \ref{theo:covariances}. For categorical data, Post-randomization Method (PRAM) leads to unbiased estimates of univariate moments, see, e.g. \cite{shlomo2008protection}, and also seems applicable to asymptotic analysis. However, PRAM operates on a single variable at a time and more elaborate techniques are required to preserve the dependency structures between multiple categorical variables. Finally, investigating modification constraints, such as selectively altering the first $j$ columns of $U_n$, presents a compelling research area, especially for high-dimensional datasets.




\begin{credits}
\subsubsection{\ackname}
The study was supported by the Finnish Cultural Foundation (grant 00220801) and the Research Council of Finland (grants 347501 and 353769). The authors would like to thank two anonymous reviewers for their valuable comments.

\subsubsection{\discintname}
The authors have no competing interests to declare that are relevant to the content of this article.

\end{credits}

\appendix

\section{Proofs of technical results}\label{sec:appen_proofs}

\begin{proof}[Proof of Theorem \ref{theo:means}]
    We begin with the $\mathcal{P}$-SA. Denoting $Y_n :=  X_{n, \mathcal{P}}$, we have $\Bar{y}_n = \frac{1}{n} Y_n' 1_n = \frac{1}{n} V_n D_n U_{n0}' 1_n + \Bar{x}_n$.
    Observing now that $U_{n0}' 1_n = U_{n}' 1_n$ and that $V_n D_n U_{n}' 1_n = 0$, the desired result follows from the central limit theorem.

    For the $\mathcal{J}$-SA, we first consider the case where the data come from $\mathcal{N}_p(0, \Lambda)$ for some diagonal matrix $\Lambda$ with strictly positive and mutually distinct diagonal elements. We then write
    \begin{align}\label{eq:mean_J_1}
        \sqrt{n} \Bar{y}_n = \frac{1}{\sqrt{n}} Y_n' 1_n = \frac{1}{\sqrt{n}} V_n D_n U_{n0}' 1_n + \sqrt{n}\Bar{x}_n.
    \end{align}
    Now, the $k$th element of $D_n U_{n0}' 1_n$ equals
    \begin{align}\label{eq:mean_J_2}
        d_{nk} u_{nk}' J_{nk} 1_n = e_k' D_n U_n' J_{nk} 1_n = e_k' V_n' X_n' H_n J_{nk} 1_n.
    \end{align}
    We next show that $X_n' H_n J_{nk} 1_n/\sqrt{n}$ admits a limiting distribution and is, as such, stochastically bounded. To see this, we let the $i$th diagonal element of $J_{nk}$ be $s_{n,i} \sim \mathrm{Uniform}\{-1 ,1\}$ and write $\frac{1}{\sqrt{n}} X_n' H_n J_{nk} 1_n = \frac{1}{\sqrt{n}} \sum_{i = 1}^n x_{n, i} s_{n,i} - \sqrt{n} \Bar{x}_n \frac{1}{n} \sum_{i = 1}^n s_{n,i}$.
    Now, the first term has a limiting distribution (by CLT) and the second term is $o_p(1)$ (by CLT and Slutsky's lemma). Thus, fixing the sign of $V_n$ such that $V_n \rightarrow_p I_p$, we have, by \eqref{eq:mean_J_2} that $\frac{1}{\sqrt{n}} d_{nk} u_{nk}' J_{nk} 1_n = \frac{1}{\sqrt{n}} e_k' X_n' H_n J_{nk} 1_n + o_p(1)$.
    Using the same in \eqref{eq:mean_J_1}, we also obtain $\sqrt{n} \Bar{y}_n = \sqrt{n}\Bar{x}_n + \frac{1}{\sqrt{n}} D_n U_{n0}' 1_n + o_p(1)$.
    This implies that the $k$th element of $\sqrt{n} \Bar{y}_n$ has the expansion $(\sqrt{n} \Bar{y}_n)_k = \frac{1}{\sqrt{n}} \sum_{i = 1}^n x_{n, ik} (1 + s_{n,i}) + o_p(1)$. As $x_{n, ik} (1 + s_{n, i})$, $i = 1, \ldots, n$, are i.i.d. with mean zero and variance $2 \mathrm{Var}(x_{n, ik})$, and as $x_{n, ik} (1 + s_{n, i})$ and $x_{n, i\ell} (1 + s_{n, i})$ are uncorrelated for $k \neq \ell$, the limiting distribution of $\sqrt{n} \Bar{y}_n$ is $\mathcal{N}_p(0, 2 \Lambda)$.

    Let now $Z_n = X_n O' + 1_n \mu'$ for some $p \times p$ orthogonal matrix $O$ and $\mu \in \mathbb{R}^p$. Then,
    $ Z_{n, \mathcal{P}} = (X_{n, \mathcal{P}} -  1_n \Bar{x}_n')O' + 1_n (O \Bar{x}_n + \mu)'$, thus implying that $ \frac{1}{n} Z_{n, \mathcal{P}}' 1_n =  O \frac{1}{n}  X_{n, \mathcal{P}}' 1_n - O \Bar{x}_n + O \Bar{x}_n + \mu$ 
    and $ \sqrt{n} \left( \frac{1}{n} Z_{n, \mathcal{P}}' 1_n - \mu \right) = \sqrt{n} O \Bar{y}_n \rightsquigarrow \mathcal{N}_p(0, 2 O \Lambda O')$.
    This completes the proof for $\mathcal{J}$-SA.

    For $\mathcal{O}$-SA, the proof is similar to $\mathcal{J}$-SA, and we point out only the differences next. To see that $X_n' H_n O_{nk} 1_n/\sqrt{n}$ is $\mathcal{O}_p(1)$, we observe that $O_{nk} 1_n \sim \sqrt{n} O_{nk} e_1 \sim \sqrt{n} u_n$, where $u_n$ is uniform on the unit sphere $\mathbb{S}^{n - 1}$. Moreover, letting $\chi_n$ denote a random variable obeying a $\chi$-distribution with $n$ degrees of freedom that is independent of $u_n$, we have $b_n := \chi_n u_n \sim \mathcal{N}_n(0, I_n)$ and $\chi_n/\sqrt{n} \rightarrow_p 1$. Consequently, denoting a generic column of $X_n$ by $x_n \sim \mathcal{N}_n(0, \lambda I_n)$, we have $ (1/\sqrt{n}) x_n' H_n O_{nk} 1_n \sim x_n' H_n u_n = x_n' u_n - \Bar{x}_n 1_n' u_n = b_n^{-1} (1/\sqrt{n}) x_n' z_n - \sqrt{n} \Bar{x}_n \Bar{z}_n b_n^{-1} $, showing that $X_n' H_n O_{nk} 1_n/\sqrt{n}$ is stochastically bounded. Hence,
    \begin{align*}
        (\sqrt{n} \Bar{y}_n)_k = ( \sqrt{n}\Bar{x}_n )_k + b_n^{-1} \frac{1}{\sqrt{n}} e_k' X_n' z_n + o_p(1) = \frac{1}{\sqrt{n}} \sum_{i = 1}^n x_{n, ik} (1 + z_{n, i}) + o_p(1).
    \end{align*}
    The limiting distribution of this is $\mathcal{N}(0, 2 \mathrm{Var}(x_{n, ik}))$, and rest of the proof follows similarly as for $\mathcal{J}$-SA.
\end{proof}

    

The following lemma is a direct consequence of Proposition 4.1 in \cite{viana2001covariance}.

\begin{lemma}\label{lem:moments_permutation}
   Let $P$ be uniformly random $n \times n$ permutation matrix. Then, $\mathrm{E}(P) = (1/n) 1_n 1_n'$, $\mathrm{E}\{ \mathrm{tr}(P^2) \} = 2$ and $\mathrm{E}[\{ \mathrm{tr}(P) \}^2] = 2$.
\end{lemma}


\begin{proof}[Proof of Theorem \ref{theo:covariances}]
    Starting again with the case of $\mathcal{N}_p(0, \Lambda)$-distributed data and using the notation of the proof of Theorem \ref{theo:means}, we have
    \begin{align*}
        \sqrt{n} (S_{n, \mathcal{P}} - \Lambda) = \sqrt{n} \left( \frac{1}{n} Y_n' Y_n - \Lambda - \Bar{y}_n  \Bar{y}_n' \right) = \sqrt{n} \left( \frac{1}{n} Y_n' Y_n - \Lambda \right) + o_p(1),
    \end{align*}
    as $\sqrt{n} \Bar{y}_n = \mathcal{O}_p(1)$. From Theorem \ref{theo:means} we have $V_n D_n U_{n0}' 1_n \Bar{x}_n/\sqrt{n} = o_p(1)$, giving
    \begin{align}\label{eq:cov_P_1}
    \begin{split}
        \sqrt{n} (S_{n, \mathcal{P}} - \Lambda) =& \sqrt{n} \left( \frac{1}{n} V_n D_n U_{n0}' U_{n0} D_n V_n' - \Lambda \right) + o_p(1), \\
        =& \sqrt{n} \left( \frac{1}{n} X_n' H_n X_n - \Lambda \right) + \frac{1}{\sqrt{n}} V_n D_n G_n D_n V_n' + o_p(1), \\
        =& \sqrt{n} \left( \frac{1}{n} X_n' X_n - \Lambda \right) + \frac{1}{\sqrt{n}} V_n D_n G_n D_n V_n' + o_p(1),
    \end{split}
    \end{align}
    where $G_n \in \mathbb{R}^{p \times p}$ has $g_{n, kk} = 0$, $g_{n, k\ell} = u_{nk}' T_{nk}' T_{n\ell} u_{n \ell}$, where $T_{nk}$ are either permutation, sign-change or orthogonal matrices, depending on the type of SA. We next show that, for $k \neq \ell$, the quantity $f_{n, k\ell} := d_{nk} d_{n\ell} g_{n, k\ell}/\sqrt{n}$ is $\mathcal{O}_p(1)$.
    
    For $\mathcal{P}$-SA, this is seen by writing $ \sqrt{n} f_{n, k\ell} = e_k' D_n U_n' P_{nk}' P_{n\ell} U_n D_n e_\ell = e_k' V_n' X_n' H_n P_{nk}' P_{n\ell} H_n X_n V_n e_\ell 
        = e_k' V_n' (X_n' P_{nk}' P_{n\ell} X_n) V_n e_\ell - n e_k' V_n' (\Bar{x}_n \Bar{x}_n') V_n e_\ell $,
    giving $ f_{n, k\ell} = n^{-1/2} e_k' V_n' (X_n' P_{nk}' P_{n\ell} X_n) V_n e_\ell + o_p(1) = n^{-1/2} x_{n, k}' P_{nk}' P_{n\ell} x_{n, \ell} + o_p(1) $,
    where $x_{n, k}$ is the $k$th column of $X_n$. For the second equality, we have used $ \frac{1}{\sqrt{n}} X_n' P_{nk}' P_{n\ell} X_n = \mathcal{O}_p(1) $
    which holds for the off-diagonal elements by the independence of $x_{n, k}$ and $x_{n, \ell}$. For the diagonal elements, the boundedness is equivalent to $x_n' P_n x_n/\sqrt{n} = \mathcal{O}_p(1)$ when $x_n \sim \mathcal{N}_n(0, I_n)$ and $P_n$ is a uniformly random $n \times n$ permutation matrix. Since $\mathrm{E}(x_n' P_n x_n) = 0$, by \cite[Theorem 14.4-1]{bishop2007discrete},
    \begin{align}\label{eq:bishop_2007}
        \frac{1}{\sqrt{n}} x_n' P_n x_n = \frac{1}{\sqrt{n}} \mathrm{SD} (x_n' P_n x_n) \cdot \mathcal{O}_p(1).
    \end{align}
    Now, $\mathrm{E}(x_n x_n' P_n x_n x_n'|P_n) = \mathrm{tr}(P_n) I_n + P_n + P_n'$, giving, using Lemma \ref{lem:moments_permutation}, that $\mathrm{Var}(x_n' P_n x_n) = \mathrm{E}[\{ \mathrm{tr}(P_n) \}^2] + \mathrm{E} \{ \mathrm{tr}(P_n^2) \} + n = n + 4$.
    Plugging in to \eqref{eq:bishop_2007}, the stochastic boundedness of $X_n' P_{nk}' P_{n\ell} X_n/\sqrt{n}$ and $f_{n, k\ell}$ now follows.

    For $\mathcal{J}$-SA, the reasoning is similar but uses the facts that $\frac{1}{n} X_n' J_{nk} J_{n\ell} 1_n = o_p(1)$ and $\frac{1}{n} 1_n' J_{nk} J_{n\ell} 1_n = o_p(1)$
    where the former holds because $J_{n\ell} J_{nk} X_n \sim X_n$. Furthermore, instead of Lemma~\ref{lem:moments_permutation}, we use $\mathrm{tr}(J_n^2) = n$ and $\mathrm{E}[\{ \mathrm{tr}(J_n) \}^2] = n$.

    For $\mathcal{O}$-SA, we similarly have that $\frac{1}{n} X_n' O_{nk} O_{n\ell} 1_n = o_p(1)$ and $\frac{1}{n} 1_n' O_{nk} O_{n\ell} 1_n = o_p(1)$,
    since $O_{n\ell} O_{nk} X_n \sim X_n$ and since, reasoning as in the proof of Theorem \ref{theo:means}, we have $\frac{1}{n} 1_n' O_{nk} O_{n\ell} 1_n = \frac{1}{n} z_{n1}' z_{n2} + o_p(1)$,
    where $z_{n1}, z_{n2} \sim \mathcal{N}_n(0, I_n)$ are independent of each other. Also, instead of Lemma \ref{lem:moments_permutation}, we use for $\mathcal{O}$-SA the bounds $\mathrm{tr}(O_n^2) \leq n$ and $\mathrm{E}\{ ( P_{n, 11} + \cdots + P_{n,nn} )^2 \} \leq \{ [\mathrm{E}(P_{n, 11}^2)]^{1/2} + \cdots + [\mathrm{E}(P_{n, nn}^2)]^{1/2} \}^2$,
    together with $\mathrm{E}(P_{n, ii}) = 0$ and $\mathrm{E}(P_{n, ii}^2) = 1/n$.

    Thus, for every form of SA, by \eqref{eq:cov_P_1}, we have $ \sqrt{n} (S_{n, \mathcal{P}} - \Lambda) = \sqrt{n} ( \frac{1}{n} X_n' X_n - \Lambda ) + \frac{1}{\sqrt{n}} D_n G_n D_n + o_p(1) $. Writing out the individual elements of $B_n := \sqrt{n} (S_{n, \mathcal{P}} - \Lambda)$, the diagonal and off-diagonal thus read
    \begin{align}
        b_{n, kk} &= \sqrt{n} \left( \frac{1}{n} x_{n, k}' x_{n, k} - \lambda_k \right) + o_p(1), \label{eq:bnkk} \\
        b_{n, k\ell} &=  \frac{1}{\sqrt{n}} x_{n, k}' (I_n + T_{nk}' T_{n\ell}) x_{n, \ell} + o_p(1). \label{eq:bnkl}
    \end{align}
    What remains now is to show that $b_{n, kk}, b_{n, k\ell}$, $k, \ell = 1, \ldots, p$, $k \neq \ell$, have a limiting joint normal distribution with the desired covariance matrix. We begin by computing the covariance matrix. For $\mathcal{P}$-SA, direct computation gives $\mathrm{Var}(b_{n, kk}) = 2 \lambda_k^2$ and $\mathrm{Var}(b_{n, k\ell}) = 2 \lambda_k \lambda_\ell \{ 1 + (1/n) \mathrm{E} \mathrm{tr}(P_{nk}' P_{n\ell}) \} = 2 (1 + 1/n) \lambda_k \lambda_\ell = 2 \lambda_k \lambda_\ell + o(1)$,
    where we have used the fact that $\mathrm{E}(P_{nk}) = 1_n 1_n'/n$. Similarly, we obtain the following non-zero covariances (rest of the covariances are zero): $\mathrm{Cov}(b_{n, k \ell}, b_{n, \ell k}) = \mathrm{Var}(b_{n, k\ell}) = 2 \lambda_k \lambda_\ell + o(1)$.
    Analogous computation reveals that the same asymptotic variances and covariances are obtained also for $\mathcal{J}$-SA and $\mathcal{O}$-SA. Consequently, in each case, the limiting covariance matrix of $\sqrt{n} \mathrm{vec} (S_{n, \mathcal{P}} - \Lambda)$ is $(\Lambda^{1/2} \otimes \Lambda^{1/2})[2 I_{p^2} + 2 K_{p,p} - 2 \mathrm{diag} \{ \mathrm{vec}(I_p) \} ](\Lambda^{1/2} \otimes \Lambda^{1/2}) $.
    The form for a general normal distribution now follows with equivariance arguments as in the proof of Theorem \ref{theo:means}. That the moments of the limiting distribution actually are the limits of the moments we computed earlier, is guaranteed by showing that $b_{n, kk}$ and $b_{n, k\ell}$ are uniformly integrable. E.g., for $b_{n, k\ell}$, this can be done by denoting $A := (1/\sqrt{n}) x_{n, k}' x_{n, \ell}$ and $B := (1/\sqrt{n}) x_{n, k}' T_{nk}' T_{n\ell} x_{n, \ell}$, using Young's inequality to bound $\mathrm{E}\{(A + B)^4\} \leq a \mathrm{E}(A^4) + b \mathrm{E}(B^4)$ for some scalars $a, b \in \mathbb{R}$ and observing that  $T_{nk} x_{n, k} \sim \mathcal{N}_n(0, I_n)$ for each form of SA. The boundedness of $\mathrm{E}(A^4)$ and $\mathrm{E}(B^4)$ now follows from the moments of the Gaussian distribution.

    Next, we show that the limiting distribution exists. By the equivariance properties of the multivariate normal, it is sufficient to consider only the case $\Lambda = I_p$. We begin with $\mathcal{O}$-SA. Letting $Q_{n1}, \ldots, Q_{np}$ be random orthogonal matrices uniform from the Haar distribution (and independent of all our other random variables), we have $x_{n,k} \sim Q_{nk} x_{n, k}$. Hence, \eqref{eq:bnkk} and \eqref{eq:bnkl} can be written as $b_{n, kk} = \sqrt{n} \{ (1/n) x_{n, k}' x_{n, k} - 1 \} + o_p(1)$ and $b_{n, k\ell} = (1/\sqrt{n}) x_{n, k}' (Q_{nk}' Q_{n\ell} + R_{nk}' R_{n\ell}) x_{n, \ell} + o_p(1)$
    where $R_{nk} := O_{nk} Q_{nk}$ is Haar-uniformly distributed orthogonal matrix and independent of $Q_{nk}$. We then decompose $x_{n,k} = \chi_{nk} u_{nk}$ where $\chi_{nk} \sim \chi_n$ and $u_{nk}$ is uniform on the unit sphere and let $\chi_{nk1}, \chi_{nk2}$ be independent $\chi_n$-random variables. Using these we write $b_{n, kk}$ and $b_{n, k\ell}$ as $ b_{n, kk} = \sqrt{n} \{ (1/n) (\chi_{nk} u_{nk})' (\chi_{nk} u_{nk}) - 1 \} + o_p(1) $ and as
   \begin{align*}
        b_{n, k\ell} &= \frac{\chi_{nk}\chi_{n\ell}}{\chi_{nk1}\chi_{n\ell 1}} \frac{1}{\sqrt{n}} (\chi_{nk1} Q_{nk} u_{nk})' (\chi_{n\ell1} Q_{n\ell} u_{n\ell})  \\
        &+ \frac{\chi_{nk}\chi_{n\ell}}{\chi_{nk2}\chi_{n\ell 2}} \frac{1}{\sqrt{n}} (\chi_{nk2} R_{nk} u_{nk})' (\chi_{n\ell2} R_{n\ell} u_{n\ell}) + o_p(1).
    \end{align*}
    Now, above $(\chi_{nk}\chi_{n\ell})/(\chi_{nk1}\chi_{n\ell 1}) \rightarrow_p 1$ and $(\chi_{nk}\chi_{n\ell})/(\chi_{nk2}\chi_{n\ell 2}) \rightarrow_p 1$. Moreover, $\chi_{nk} u_{nk}$, $\chi_{nk1} Q_{nk} u_{nk}$ and $\chi_{nk2} R_{nk} u_{nk}$ are independent of each other and all obey the $\mathcal{N}_n(0, I_n)$-distribution. Hence, by CLT and Slutsky's lemma, the desired joint limiting normal distribution for $b_{n, kk}$ and $b_{n, k\ell}$, $k, \ell = 1, \ldots, p$, $k \neq \ell$, is obtained in the case of $\mathcal{O}$-SA.

    For $\mathcal{J}$-SA, joint limiting normality follows from the multivariate CLT, as $J_{n1}, \ldots, J_{np}$ do not ``mix'' the observations and retain their i.i.d. nature. 

    What is thus left is $\mathcal{P}$-SA. In that case, we use the Cram{\'e}r-Wold device combined with the CLT for local dependency neighbourhoods \cite[Theorem 3.6]{ross2011fundamentals}. We demonstrate the proof below in the case $p = 2$, the general case following analogously (but with more cluttered notation). That is, the claim holds for $p = 2$ once we show that $ S_n := a_{11} \sqrt{n} ( \frac{1}{n} x' x - 1 ) + a_{12} \frac{1}{\sqrt{n}} x' (I_n + P) y + a_{22} \sqrt{n} ( \frac{1}{n} y' y - 1 ) $ admits a limiting normal distribution where $a_{11}, a_{12}, a_{22}$ are arbitrary constants, $P$ is a uniformly random $n \times n$ orthogonal matrix and $x, y \sim \mathcal{N}_n(0, I_n)$ are independent of each other and $P$. Our objective is to use \cite[Theorem 3.6]{ross2011fundamentals}, where we will take (using their notation) $X_i = a_{11} (x_i^2 - 1)/\sqrt{n} + a_{12} x_i (y_i + y_{\delta(i)})/\sqrt{n} + a_{22} (y_i^2 - 1)/\sqrt{n}$, where $\delta$ is the permutation mapping corresponding to the permutation matrix $P$, and $\sigma^2 = \mathrm{Var}(S_n) = 2 a_{11} + 2 a_{12} (1 + 1/n) + 2 a_{22}$.

    Inspection of the proof of their Theorem 3.6 reveals that the dependency neighbourhood $N_i$ can be taken to be random (as they are for us), as long as their maximal degree is almost surely some finite $D$ (as it is for us, since $x_i (y_i + y_{\delta(i)})$ depends on maximally two other terms), when specific modifications to the proof/result are done: (i) The first term on the RHS of the statement of Theorem 3.6 can be kept as such, but when deriving it, the sums of the form ``$\sum_{j \in N_i}$'' have to be replaced with ``$\sum_{j = 1}^n \mathbb{I}(j \in N_i)$'', since $N_i$ are random. (ii) In their formula (3.12), we decompose the term $\mathrm{Var} ( \sum_{i = 1}^n \sum_{j \in N_i} X_i X_j )$ as
    \begin{align*}
        & \mathrm{E} \left\{ \mathrm{Var} \left( \sum_{i = 1}^n \sum_{j \in N_i} X_i X_j \mid P \right) \right\} + \mathrm{Var} \left\{ \mathrm{E} \left( \sum_{i = 1}^n \sum_{j \in N_i} X_i X_j \mid P \right) \right\} \\
        \leq& 14 D^3 \sum_{i = 1}^n \mathrm{E}(X_i^4) + \mathrm{Var}\{ \mathrm{Var}(S_n \mid P) \},
    \end{align*}
    where the inequality uses the variance bound derived in the proof of \cite[Theorem 3.6]{ross2011fundamentals} (for our conditioned-upon dependency neighborhoods) and the second term uses the fact that $\mathrm{Var}(\sum_{i = 1}^n X_i \mid P) = \mathrm{E} ( \sum_{i = 1}^n \sum_{j \in N_i} X_i X_j \mid P )$.

    Now, as $D$ is fixed, as $\sigma^2$ approaches a non-zero constant when $n \rightarrow \infty$, and as the $X_i$ are identically distributed, the desired claim holds once we show that
    \begin{align}\label{eq:three_remaining_things}
        \mathrm{E} |X_1|^3 = o(1/n), \quad \mathrm{E} X_1^4 = o(1/n) \quad \mbox{and} \quad \mathrm{Var}\{ \mathrm{Var}(S_n \mid P) \} = o(1).
    \end{align}

    For the third claim in \eqref{eq:three_remaining_things}, we have $\mathrm{Var}(S_n \mid P) = 2 a_{11} + 2 a_{12} \{ 1 + (1/n) \mathrm{tr}(P) \} + 2 a_{22}$.
    Hence,
    \begin{align*}
        \mathrm{Var}\{ \mathrm{Var}(S_n \mid P) \} = \frac{4 a_{12}^2}{n^2} \mathrm{Var} \{ \mathrm{tr}(P) \} \leq \frac{4 a_{12}^2}{n^2} \mathrm{E} [ \{ \mathrm{tr}(P) \}^2 ] = \frac{8 a_{12}^2}{n^2},
    \end{align*}
    where we use Lemma \ref{lem:moments_permutation}, taking care of the third claim. Then, as $X_1$ equals $1/\sqrt{n}$ times a random variable with all moments finite, $ X_1 = (1/\sqrt{n}) \{ a_{11} (x_1^2 - 1) + a_{12} x_1 (y_1 + y_{\sigma(i)}) + a_{22} (y_1^2 - 1) \} $, the first two claims in \eqref{eq:three_remaining_things} also trivially hold. This conludes the proof in the case $p = 2$ and the general case is handled analogously.
    
\end{proof}

\section{Additional figures}\label{sec:appen_fig}

Figure \ref{fig:poi_met} shows the results of the utility simulation study in Section \ref{sec:finite} for Poisson-distributed data. The top (bottom) row of Figure \ref{fig:poi_met} is analogous to Figure \ref{fig:norm_met} (Figure \ref{fig:norm_viol}).

\begin{figure}[!h]
\includegraphics[width=\textwidth]{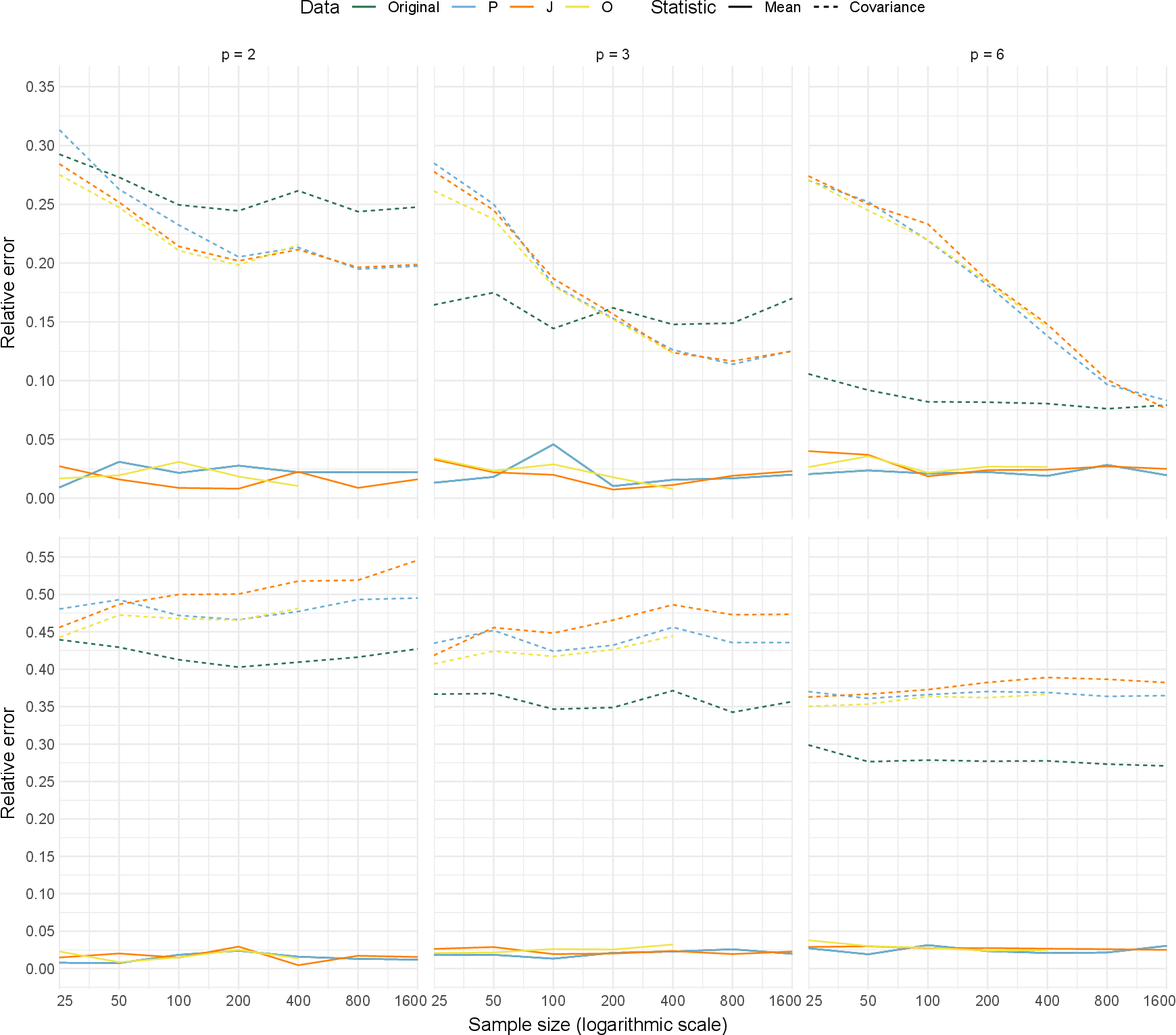}
\caption{Relative error of the empirical covariance matrices of sample mean (solid line) and sample covariance (dotted line) compared to their asymptotic covariance matrices for $\mathcal{P}$-SA (blue), $\mathcal{J}$-SA (orange), and $\mathcal{O}$-SA (yellow). In the top row, the original data (green/darkest) is sampled from a Poisson distribution meeting Assumption~\ref{assu:eigenvalues} while on the second row the assumption is violated. Sample size is on a logarithmic scale.}
\label{fig:poi_met}
\end{figure}

\bibliographystyle{splncs04}
\bibliography{PSD_2024}

\end{document}